\documentclass[journal,comsoc]{IEEEtran_kit}
\usepackage{cite}
\usepackage{float}
\usepackage{graphicx}
\usepackage{amsmath}
\usepackage{times}
\usepackage{graphicx}
\usepackage{latexsym}
\usepackage{graphicx}
\usepackage{bm}
\usepackage{amssymb}
\usepackage[center]{caption2}
\usepackage{stfloats}
\usepackage{cases}
\usepackage{array}
\usepackage{setspace}
\usepackage{fancyhdr}
\usepackage{color}
\usepackage{subfigure}
\usepackage{citesort}
\usepackage{algorithm}
\usepackage{algpseudocode}
\usepackage{algorithmicx}

\renewcommand{\QED}{\QEDopen}
\newtheorem{theorem}{Theorem}

\newtheorem{remark}{Remark}
\renewcommand{\baselinestretch}{0.97}

\newcaptionstyle{mystyle2}{%
\captionlabel $.$ \, \captiontext \par} \captionstyle{mystyle2}
\IEEEoverridecommandlockouts

\begin{document}
\title{{Symbol Detection of Ambient Backscatter Systems with Manchester Coding}}
\author{Qin Tao, Caijun Zhong, Hai Lin, and Zhaoyang Zhang
\thanks{Qin Tao, Caijun Zhong and Zhaoyang Zhang are with the Institute of Information and Communication Engineering, Zhejiang University, China, and the Zhejiang Provincial Key Laboratory of Information Processing, Communication and Networking, (email:caijunzhong@zju.edu.cn).}
\thanks{H. Lin is with the Department of Electrical and Information Systems, Osaka Prefecture University, Osaka 599-8531, Japan (email: lin@eis.osakafu-u.ac.jp).}
}


\maketitle

\begin{abstract}
Ambient backscatter communication is a newly emerged paradigm, which utilizes the ambient radio frequency (RF) signal as the carrier to reduce the system battery requirement, and is regarded as a promising solution for enabling large scale deployment of future Internet of Things (IoT) networks. The key issue of ambient backscatter communication systems is how to perform reliable detection. In this paper, we propose novel encoding methods at the information tag, and devise the corresponding symbol detection methods at the reader. In particular, Manchester coding and differential Manchester coding are adopted at the information tag, and the corresponding semi-coherent Manchester (SeCoMC) and non-coherent Manchester (NoCoMC) detectors are developed. In addition, analytical bit error rate (BER) expressions are characterized for both detectors assuming either complex Gaussian or unknown deterministic ambient signal. Simulation results show that the BER performance of unknown deterministic ambient signal is better, and the SeCoMC detector outperforms the NoCoMC detector. Finally, compared with the prior detectors for ambient backscatter communications, the proposed detectors have the advantages of achieving superior BER performance with lower communication delay.
\end{abstract}


\section{Introduction}
Internet of Things (IoT) is one of the fastest growing sectors in the wireless industry, and is gaining considerable interests from both the industry and academia \cite{L.Atzori}. A distinctive feature of IoT networks is the huge number of devices to be connected, which poses significant challenges for the practical deployment of IoT networks. For instance, due to the sheer volume of the devices, individual device should be extremely low cost. Routine maintenance procedures such as battery replacement incur overwhelming overheads. Therefore, how to address these issues is of critical importance.

One promising technology to tackle the above challenges is the backscatter communication, where the backscatter device reflects rather than generates the radio frequency (RF) signal for information transmission, so it can be made battery-free and inexpensive \cite{K.Han}. The backscatter communication has already been adopted in several commercial systems, and among which, the most notable one is the radio frequency identification (RFID) system \cite{C.He2,C.He4}. However, the communication range of RFID system is very limited due to the ``power-up link'' \cite{DMDobkin}, making it inapplicable for IoT systems. Responding to this, the work \cite{30} proposed a bistatic scatter radio which detaches the carrier emitter from the reader. By doing so, long range communication between the tag and reader can be achieved.

The bistatic scatter radio requires a dedicated carrier emitter, which may not be available or difficult to deploy in certain environments. Motivated by this, the authors in \cite{3} proposed the concept of ambient backscatter system, which utilizes the ambient RF signals from surrounding environments such as TV and cellular to establish reliable communication between the tag and reader. Later, the novel ambient WiFi backscatter was developed in \cite{Kellogg2015Wifi,BackFi,wifi}, which bridges backscatter devices with the Internet through commercial receiver. The idea was so intriguing that substantial interests have been drawn from the academia.
To improve the throughput of the ambient backscatter communication system, a multi-antenna cancellation and a three states coding scheme have been proposed in \cite{18} and \cite{YLiu}, respectively.
Further, the performance of the ambient backscatter in legacy systems was analyzed in \cite{DDarsena}, where it was shown that the backscatter transmission can even improve the performance of legacy system in certain case.
More recently, the work \cite{HoangDT} showed that the network performance can be improved taking advantage of the ambient backscatter communications.

Parallel with the pursuit of high throughput ambient backscatter communication systems, significant efforts have been devoted to seek efficient and reliable symbol detection methods. {Unlike conventional RFID systems, the received signal at the reader is corrupted by the unknown and modulated ambient signals, which makes reliable detection a challenging problem.} In a recent work \cite{1}, the authors proposed a semi-coherent energy detector and analytically characterized the achievable bit error rate (BER) performance. Later in \cite{2,5}, differential coding based non-coherent detectors were proposed.

There are several common features of the proposed detectors in \cite{1,2,5}. The first is that a decision threshold needs to be estimated, which consumes precious time and energy resources. The second is that the symbol decoding starts after the completion of the estimation process, which results in communication delay. The third is that all the detectors assume that the information bits ``0'' and ``1'' are equally probable. However, in practice, the distribution of the information bits are unknown and may fluctuate over the time, which limits the applicability of the proposed detectors.

Motivated by the above key observations, in this paper, we develop {new coding schemes for ambient backscatter communication systems} in an effort to circumvent the above issues, and devise the corresponding detection methods.
In particular, we propose to use Manchester code and differential Manchester code to encode the original information bits at the tag. With the proposed coding scheme, each information bit corresponds to a level transition. Then, {semi-coherent Manchester (SeCoMC) and non-coherent Manchester (NoCoMC)} detectors are devised, which eliminate the requirement of estimating decision threshold, and enable immediate symbol-by-symbol detection. Moreover, for both detectors, the achievable BER performance is derived in closed-form {for both the complex Gaussian signal and deterministic signal.} The outcomes of the work indicate that it is desirable to use deterministic ambient signal in terms of BER performance. Also, the proposed Manchester coding framework yields better BER performance compared with prior works \cite{1,2,5}, especially when the original information bits are unequally distributed.

The remainder of the paper is organized as follows. Section II describes in detail the system model, while Section III presents the SeCoMC detector. Section IV deals with the NoCoMC detector. Numerical results and discussions are presented in Section V. Finally, Section VI concludes the paper and summarizes the key outcomes of this paper.

{\it Notations:} Scalars are lowercase letters, while vectors and matrices are boldfaced letters. We use $h^*$ and $|h|$ to denote the conjugate and absolute value of complex number $h$, respectively. Also, $CN(\mu,\sigma^2)$ denotes the complex Gaussian distribution with mean $\mu$ and variance $\sigma^2$, $\chi _{\nu}^2$ denotes the central chi-squared distribution with $\nu$ degrees of freedom (DOF), $\chi_{\nu}^{'2}(\lambda)$ denotes the non-central chi-squared distribution with $\nu$ DOF and noncentral parameter $\lambda$. The Hermitian and determinant of matrix \textbf{A} are denoted by \textbf{A}$^H$, and det(\textbf{A}), respectively. Also, \textbf{I}$_N$ denotes the identity matrix of size $N$, and $||\textbf{y}||$ denotes the Euclidean norm of vector \textbf{y}.

\section{System Model}
We consider the elementary ambient backscatter system as in \cite{1,2,5}, which consists of three nodes, namely the ambient RF source S, the reader R, and the tag T, as {shown} in Fig.\ref{fig:fig111}. We assume the frequency flat block fading scenario, such that all channels {remain unchanged} within each coherence interval, but vary independently in different coherence intervals.
\begin{figure}[htbp]
\centering
\includegraphics[width=2.5in]{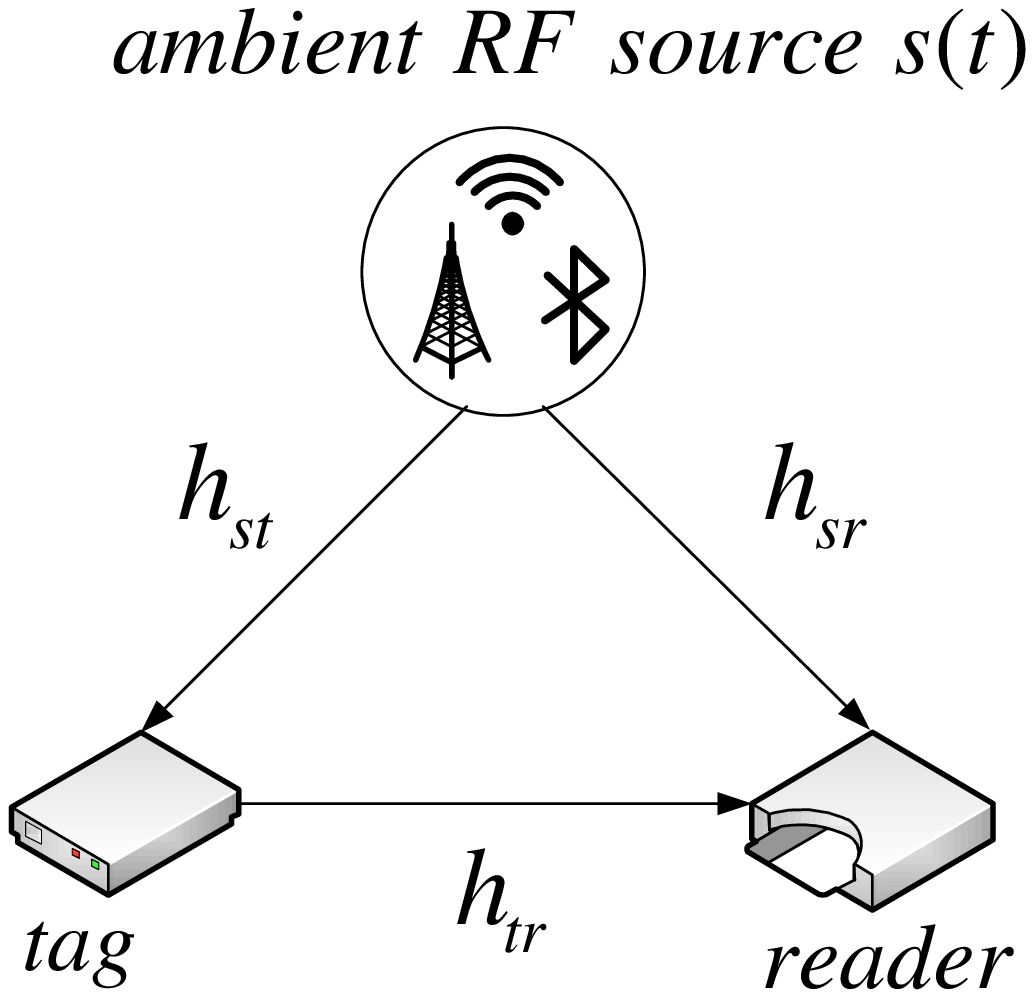}
\caption{Three-node ambient backscatter system model}
\label{fig:fig111}
\end{figure}
\subsection{Information transmission}
In ambient backscatter systems, the transmission of the binary digit $d$ of the tag to the reader is accomplished by the choice of whether to backscatter the incident ambient signal. Specifically, the digits ``1'' and ``0'' are associated with the backscattering and non-backscattering state, respectively. Since the tag {transmits} at a lower rate than the ambient RF signal, the binary digit $d$ remains unchanged for some consecutive $s(t)$. Mathematically, the backscattered signal $s_b(t)$ can be expressed as
\begin{align}
s_b(t)=\eta  h_{st}d s(t),
\end{align}
where $\eta$ is the reflection coefficient of the tag, $h_{st}$ denotes the channel coefficient between the ambient RF source and tag, $d$ is the binary digit transmitted by the tag, which will be elaborated in section \ref{sc}, $s(t)$ is the RF signal from the ambient RF source {and will be elaborated in section \ref{sb}}.

Since the reader can overhear the signals from both the ambient RF source and tag, the received signals at the reader can be expressed as $y(t)$ \footnote{Strictly speaking, the signal received by the reader between the tag and ambient RF source may exist a time delay. However, such delay is negligible, since the tag and the reader are relatively close \cite{3,1,2}.}
\begin{align} \label{g2}
y(t) &= {h_{sr}}s(t) + h_{tr}s_b(t) + w(t)\nonumber\\
&=\left[{h_{sr}} + \eta {h_{tr}}{h_{st}}d\right]s(t) + w(t),
\end{align}
where  $h_{sr}$ and $h_{tr}$ denote the channel coefficients of the ambient RF source to reader and tag to reader channels, respectively. Also, $w(t)$ is the zero-mean addictive white Gaussian noise (AWGN) with variance $N_w$, i.e., $w(t)\sim CN(0,N_w)$.

Sampling each $d$ interval at the signal Nyquist rate with $N$ such that the adjacent samples are uncorrelated, and denoting the discrete sample vector at the reader as ${\bf{y}} = \{y[1], \cdots y[n],\cdots ,y[N]\}$, then (\ref{g2}) can be reformulated as
\begin{align} \label{g3}
y[n]{\rm{ = }}\left\{ {\begin{array}{*{20}{c}}
{{h_0}s[n] + w[n],\qquad d=0},\\
{{h_1}s[n] + w[n],\qquad d=1},
\end{array}} \right.
\end{align}
where ${h_0} \triangleq {h_{sr}},{h_1} \triangleq {h_{sr}} + \eta {h_{tr}}{h_{st}}$.

\subsection{Manchester coding and differential Manchester coding} \label{sc}

Instead of transmitting the original information bits, we propose to adopt Manchester coding at the tag, in an effort to overcome the implementation issues of the detection schemes proposed in \cite{1,2,5}. To make the paper self-contained, we now provide a brief introduction of the basic idea of Manchester coding.

\subsubsection{Manchester coding}
The Manchester code is a very simple block code that maps ``0'' and ``1'' into ``01'' or ``10'', which has been widely used in passive RFID \cite{44,45,46,47,43,MSimon}. In this paper, we adopt the IEEE 802.3 standard convention for Manchester coding, where the original binary symbol ``0'' is represented by ``10'' and ``1'' is represented by ``01'', as depicted in {Fig. \ref{fig:fig3}, where ${\bar{d}}_{k}^a$ and $\bar{d}_{k}^b$} denote the first and second half of the Manchester code associated with the $k$-th ($k\in \mathbb{N}^+$) original binary symbol $d_k$, respectively.

\subsubsection{Differential Manchester coding}
The differential Manchester coding is a modification of Manchester coding, which is coded by the following rule: each bit is represented by the presence or absence of a change compared with the previous bit, i.e., no change denotes ``0'' while change denotes ``1'', as depicted in Fig. \ref{fig:fig3}, where {$\hat{d}_{k}^a$ and $\hat{d}_{k}^b$} denote the first and second half of the differential Manchester code associated with the $k$-th original binary symbol $d_k$, respectively.

\begin{figure}[htbp]
\centering
\includegraphics[width=3.5in]{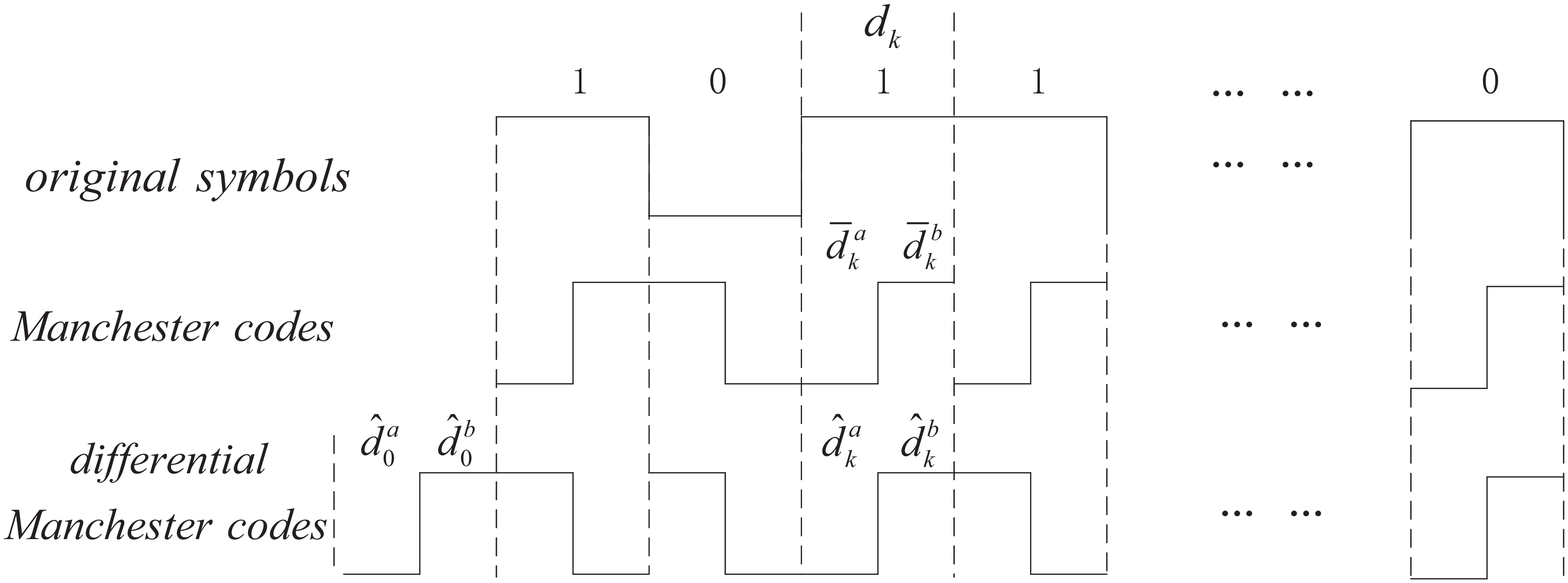}
\caption{Manchester and differential Manchester coding}
\label{fig:fig3}
\end{figure}

\begin{remark}\label{r3}
Sampling each $\bar{d}$ or $\hat{d}$ interval with rate N is equivalent to sample the original symbol with rate $2N$. \footnote{{We use $\bar{d}$ and $\hat{d}$ to denote an arbitrary symbol after Manchester coding and differential Manchester coding, respectively. This same convention applies to the notations defined in the ensuing sections such as ${\bf{\bar{y}}}$, $\bar{Z}$, ${\bf{\hat{y}}}$ and $\bar{Z}$.}}
\end{remark}
\subsection{Ambient RF Signals} \label{sb}

Depending on the communication environment, the ambient RF signals may come from a variety of ambient RF sources, such as TV, Radio, cellular network, Wi-Fi and Bluetooth transmissions. Hence, the ambient RF signals may take different forms. In the paper, we consider two typical ambient RF signals.

\subsubsection{The complex Gaussian ambient signal}
In a complex communication environment, the ambient RF signal can be the combination of many random signals. Invoking the central limit theorem, it is reasonable to model it as a zero-mean circular symmetric complex Gaussian random variable, i.e.,
\begin{align}
s[n]\sim CN(0,P_s),
\end{align}
where the $P_s$ is average power.
\subsubsection{The unknown deterministic ambient signal}
In practice, the ambient RF signal always have specific modulation mode, like $FSK$, $MSK$, $GMSK$, $PSK$, $QAM$, $OFDM$. Thus, we also analyze the detection performance in case of the unknown deterministic ambient signal.

\section{Semi-coherent Manchester detector}
In this section, we propose a SeCoMC detector and present a detailed analysis of the achievable BER performance for both complex Gaussian and deterministic RF signals. We start with the complex Gaussian signal.

\subsection{Complex Gaussian Signal}
Let $\bar H_{0}$ and $\bar H_{1}$ be the hypotheses associated with $\bar d$=0 and $\bar d$=1, respectively. When $s[n]$ and $w[n]$ are zero-mean circular symmetric complex Gaussian random variables, it can be easily observed that the received signal vector ${\bf{\bar y}}$ is a complex Gaussian vector with
\begin{align}
{\bf{\bar y}}{\rm{ \sim }}\left\{ {\begin{array}{*{20}{c}}
{CN(0,\sigma _{0}^2\textbf{I}_N),\qquad \bar H_{0}},\\
{CN(0,\sigma _{1}^2\textbf{I}_N),\qquad \bar H_{1}},
\end{array}} \right.
\end{align}
where
\begin{align}
\sigma _{0}^2 = |{h_0}{|^2}{P_s} + {N_w},\qquad \sigma _{1}^2 = |{h_1}{|^2}{P_s} + {N_w}.
\end{align}
Hence, the probability density function (PDF) of ${\bf{\bar y}}$ under hypothesis $\bar H_{i}$, {where $i \in \{0,1\}$}, is given by \cite{35,36}
\begin{align} \label{g01}
\text{Pr}({{\bf{\bar y}}}|{\bar H_{i}}) &= \frac{1}{{{{(2\pi )}^N}\det ({\sum _i})}}{e^{ - \frac{{{{{\bf{\bar y}}}}{\sum _i}^{ - 1}{{\bf{\bar y}}}^H}}{2}}}\nonumber\\
&= \frac{1}{{{{(2\pi )}^N}{{(\frac{1}{2}\sigma _{i}^2)}^N}}}{e^{ - \frac{{||{{{\bf{\bar y}}}}|{|^2}}}{{\sigma _{i}^2}}}},
\end{align}
where ${\sum _i} = \frac{1}{2}E[{{\bf{\bar y}}}^H{{{\bf{\bar y}}}}] = \frac{1}{2}{\sigma _{i}^2}\bf{I}_N$.

To detect the $k$-th original symbol $d_k$, we adopt the maximum likelihood (ML) principle as in \cite{1}. Let $H_0$ denote the hypothesis of $d_k=0$ and $H_1$ denote the hypothesis of $d_k=1$, with probability $q$ and $1-q$, respectively. With Manchester coding, the received signal during the $k$-th symbol interval can be expressed as
\begin{align}
{{\bf{y}}_k} = {{\bf{\bar y}}_{k}^a}{{\bf{\bar y}}_{k}^b}, \quad ab \in \{01,10\}.
\end{align}
Since ${{\bf{\bar y}}_{k}^a}$ and ${{\bf{\bar y}}_{k}^b}$ are independent, the ML ratio test can be obtained as
\begin{align}\label{g8}
L({\bf{y}}_k) =  \frac{{p({{\bf{\bar y}}_{{{k}}}^a}{{\bf{\bar y}}_{{{k}}}^b}|{ H_1})}}{{p({{\bf{\bar y}}_{{{k}}}^a}{{\bf{\bar y}}_{{{k}}}^b}|{ H_0})}}
=\frac{{p({{\bf{\bar y}}_{{{k}}}^a}|{\bar H_{0}})p({{\bf{\bar y}}_{{{k}}}^b}|{\bar H_{1}})}}{{p({{\bf{\bar y}}_{{{k}}}^a}|{\bar H_{1})p({{\bf{\bar y}}_{{{k}}}^b}|{\bar H_{0}})}}}
{\begin{array}{*{20}{c}}
{{H_1}}\\
 \gtrless\\
{{H_0}}
\end{array}}
1.
\end{align}
Now, {denoting $\bar Z_{k}^{j}=||{{\bf{\bar y}}_{k}^{j}}|{|^2}$ as the received signal energy, where $j\in \{a, b\}$}, and substituting (\ref{g01}) into (\ref{g8}), after some simple algebraic manipulations, we have
\begin{align}
\frac{1}{{\sigma _{0}^2}}(\bar Z_{k}^a - \bar Z_{k}^b)
 {\begin{array}{*{20}{c}}
{{H_0}}\\
 \gtrless\\
{{H_1}}
\end{array}}
 \frac{1}{{\sigma _{1}^2}}(\bar Z_{k}^a - \bar Z_{k}^b).
\end{align}

Therefore, the decision rule of the proposed SeCoMC detector is given by
\begin{align}
\left\{ {\begin{array}{*{20}{c}}
{\bar Z_{k}^a\begin{array}{*{20}{c}}
{\begin{array}{*{20}{c}}
{{H_0}}\\
 \gtrless\\
{{H_1}}
\end{array}}
\end{array}\bar Z_{k}^b,\quad \sigma _{0}^2< \sigma _{1}^2}.\\
{\bar Z_{k}^a\begin{array}{*{20}{c}}
{\begin{array}{*{20}{c}}
{{H_0}}\\
 \lessgtr\\
{{H_1}}
\end{array}}
\end{array}\bar Z_{k}^b, \quad \sigma _{0}^2 > \sigma _{1}^2}.
\end{array}} \right.
\end{align}
{As can be readily observed, to recover the original symbol, one can simply compare the energy level of two adjacent Manchester coded symbols.{\footnote{{If $\sigma_0^2=\sigma_1^2$, the detector fails. However, such case is considered unlikely \cite{1}.}}}} {Please note, the above decision rule is significantly different from the traditional RFID with Manchester coding and the rule proposed in \cite{1}, which compares the energy of the symbol interval with some predetermined decision threshold.} Since the proposed scheme does not require the exact evaluation of the energy level as well as the decision threshold, it is much more energy efficient, and incurs less delay.

Also, to make the decision, the relationship of $\sigma_0^2$ and $\sigma_1^2$ is required. It is worth noting that, because the channel coefficients are complex, $\sigma_1^2$ is not guaranteed to be greater than $\sigma_0^2$. Since the relationship is unknown a priori, it needs to be estimated in each coherent block. In practice, this can be achieved via training. For instance, at the beginning of each {coherence interval of length $K$ symbols}, a successive number of $T$ symbols ``1'' are used to evaluate the relationship. Now define $A_t$ and $B_t$ as
\begin{align}
A_t = \frac{{\sum\limits_{t = 1}^T{{{{\bar Z}_{t}^a}}}}}{T\cdot N}\quad {\mbox{and}} \quad B_t = \frac{{\sum\limits_{t = 1}^T{{{{\bar Z}_{t}^b}}}}}{T\cdot N}.
\end{align}
Then, we can use $A_t$ and $B_t$ to approximate $\sigma _{0}^2$ and $\sigma _{1}^2$, respectively. If $A_t>B_t$, then $\sigma _{0}^2>\sigma _{1}^2$, else $\sigma _{0}^2<\sigma _{1}^2$.

We summarize the algorithm of SeCoMC detector as follows:
\begin{algorithm}[h]
\caption{SeCoMC detector}
\begin{algorithmic} [1]
\Require The received signal vectors: [$\underbrace {{{\bf{y}}_1}; \cdots {{\bf{y}}_T}}_{d = 1}; {\bf{y}}_1$; $\cdots$ ${\bf{y}}_k$; $\cdots$ ${\bf{y}}_K$]
\Ensure The detected symbols: [$d_1$, $\cdots$, $d_k$ $\cdots$, $d_K$]
\State Training phase: Evaluate the relationship of $A_t$ and $B_t$
\State For $k$ from $1$ to $K$\\
    \quad compute $\bar Z_{k}^a$=$||{\bf{\bar y}}_{k}^a||^2$, $\bar Z_{k}^b$=$||{\bf{\bar y}}_{k}^b||^2$\\
     \quad if $A_t > B_t$\\
     \quad\quad if $\bar Z_{k}^a>\bar Z_{k}^b$, then let $d_k=1$, else $d_k=0$, end if\\
     \quad else if $\bar Z_{k}^a\leq\bar Z_{k}^b$, then let $d_k=1$, else $d_k=0$, end if\\
     \quad end if \\
     end for
\State Return [$d_1$, $\cdots$, $d_k$ $\cdots$, $d_K$]
\end{algorithmic}
\end{algorithm}

We now present a detailed performance analysis on the achievable BER of the proposed SeCoMC detector, and we have the following important result:
\begin{theorem} \label{t0}
The BER of SeCoMC detector with complex Gaussian signal is given by
\begin{align}
{P_{se}^{CG}}
&={\frac{\Gamma(2N)\sigma _{n}^{2N}}{N\Gamma^2(N)\sigma _{m}^{2N}}} \cdot {_2F_1}\left(N,2N;N+1;-\frac{\sigma _{n}^{2}}{\sigma _{m}^{2}}\right),
\end{align}
where $\sigma _{n}^{2}=\min\{\sigma _{0}^{2},\sigma _{1}^{2}\}$, $\sigma _{m}^{2}=\max\{\sigma _{0}^{2},\sigma _{1}^{2}\}$, $\Gamma(x)$ denotes the gamma function and ${_2F_1}(a, b ; c ; -x)$ denotes the Gauss hypergeometric function \cite{Table}.

\begin{proof}
See Appendix \ref{ap2}.
\end{proof}
\end{theorem}

In contrast to the detector proposed in \cite{1}, which depends heavily on the distribution of $H_1$ and $H_0$, we see that the proposed detector is independent of $q$. This is a highly desirable feature since the BER performance guarantee can be ensured for arbitrary $q$. While Theorem \ref{t1} provides an efficient means for the evaluation of the BER performance, it is nevertheless difficult to reveal the impact of the key system parameters on the system performance. Therefore, we now look into the asymptotic regime, where simple expressions can be obtained.

\begin{theorem} \label{t1}
When $N$ is large, the BER of SeCoMC detector with complex Gaussian signal can be approximated as
\begin{align}\label{gap2}
\widetilde{P}_{se}^{CG} = \frac{1}{2}{\sf erfc}\left(\frac{{\sqrt N ||{h_1}{|^2} - |{h_0}{|^2}| }}{{\sqrt 2 \sqrt {{{(|{h_0}{|^2} + \frac{1}{\gamma })}^2} + {{(|{h_1}{|^2} + \frac{1}{\gamma })}^2}} }}\right),
\end{align}
where ${\sf erfc}(x)= 1 - {\sf erf}(x)$, ${\sf erf}(x) =\frac{2}{\sqrt{\pi} }\int_0^x {{e^{ - {t^2}}}} dt$ and $\gamma \triangleq \frac{{{P_s}}}{{{N_w}}}$ denotes the signal to noise ratio (SNR).
\end{theorem}
\begin{proof}
See Appendix \ref{ap2n}.
\end{proof}

Due to the monotonicity of the ${\sf erfc}$ function, it is easy to show that the BER is a decreasing function with respect to $N$, indicating that increasing the sampling rate is always beneficial. In addition, when the SNR increases, the BER performance improves. In the asymptotic high SNR regime, i.e., $\gamma\rightarrow \infty$, (\ref{gap2}) reduces to
\begin{align} \label{g20}
\widetilde{P}_{se}^{CG} \approx \frac{1}{2}{\sf erfc}\left(\frac{{\sqrt N ||{h_1}{|^2} - |{h_0}{|^2}| }}{{\sqrt 2 \sqrt {{{|{h_0}{|^4} }} + {{|{h_1}{|^4}}}} }}\right).
\end{align}

The above expression indicates that, as the SNR increases, the BER reaches an error floor, {which is determined by the sampling rate $N$ and the relative channel difference (RCD) of the path, i.e., RCD$\triangleq \frac{{||{h_1}{|^2} - |{h_0}{|^2}| }}{{\sqrt {{{|{h_0}{|^4} }} + {{|{h_1}{|^4}}}} }}$}.

We now compare the BER performance of the proposed detector with that of the semi-coherent detector proposed in \cite{1}. Recalling the BER of the semi-coherent detector proposed in \cite{1}, it can be expressed as
\begin{align} \label{gap22}
\widetilde{P}_b^{CG} = \frac{1}{2}{\sf erfc}\left(\frac{{\sqrt {N} ||{h_1}{|^2} - |{h_0}{|^2}| }}{{ {{{|{h_0}{|^2} } + {{|{h_1}{|^2} + \frac{2}{\gamma }}}} }}}\right).
\end{align}

A close observation of (\ref{gap2}) and (\ref{gap22}) reveals that, to compare $\widetilde{P}_{se}^{CG}$ and $\widetilde{P}_b^{CG}$, it is sufficient to compare the {denominators} inside the ${\sf erfc}$ functions. Hence, let us compute the difference of the square of the two {denominators}, and we have
\begin{align}
&\left(|{h_0}|^2  + |{h_1}|^2 + \frac{2}{\gamma }\right)^2\notag\\
&\qquad \quad -\left({\sqrt 2 \sqrt {{{\left(|{h_0}{|^2} + \frac{1}{\gamma }\right)}^2} + {{\left(|{h_1}{|^2} + \frac{1}{\gamma }\right)}^2}} }\right)^2\notag\\
&=-\left(|{h_0}|^2  - |{h_1}|^2 \right)^2\leq 0
\end{align}

Since ${\sf {erfc}}(x)$ is a monotonically decreasing function with respect to $x$, we have
\begin{align}\label{eqncom}
\widetilde{P}_{se}^{CG} \geq \widetilde{P}_b^{CG}.
\end{align}
Hence, theoretically, the BER performance of the proposed SeCoMC detector is inferior. {The main reason is that the semi-coherent detector in \cite{1} is based on the absolute symbol energy, while the proposed detector is based on the energy difference of the first and second half of the entire symbol interval, which may causes some information loss.} However, it is worth highlighting that $\widetilde{P}_b^{CG}$ {can only be achieved when the threshold is perfect. In practice, it needs to be estimated, which causes performance degradation due to estimation error}. As will be shown later via simulation, the proposed SeCoMC detector actually performs better in practice.

\subsection{Unknown deterministic signal}
We now investigate the performance of the SeCoMC detector with unknown deterministic signal. Without any information about the signal, it is appropriate to apply the energy detector \cite{48}. It is also worthy highlighting that, since only the energy of the signal is used for detection, the results presented here are applicable to arbitrary deterministic signals.

\begin{theorem} \label{tp}
When $N$ is large, the BER of the SeCoMC detector with unknown deterministic signal is given by
\begin{align}\label{g35}
\widetilde{P}_{se}^{UD} = \frac{1}{2}{\sf erfc}\left(\frac{\sqrt N \left||h_1|^2 - |h_0|^2\right| }{2 \sqrt {\frac{|h_0|^2 +|h_1|^2}{\gamma}+ \frac{1}{\gamma^2 } }}\right).
\end{align}
\end{theorem}

\begin{proof}
See Appendix \ref{ap3}.
\end{proof}
In the asymptotic high SNR regime, i.e., $\gamma\rightarrow \infty$, (\ref{gap2}) reduces to
\begin{align} \label{g37}
\widetilde{P}_{se}^{UD}\approx\frac{1}{2}{\sf erfc}\left(\frac{\sqrt {N\gamma} \left||h_1|^2 - |h_0|^2\right| }{2 \sqrt {|h_0|^2 +|h_1|^2}}\right).
\end{align}
As expected, we can see that the BER {performance} of unknown deterministic signal also improves when $N$ becomes large. However, unlike the complex Gaussian scenario, no error floor exists.

{Comparing the BER performance for the scenarios with complex Gaussian signal and unknown deterministic signal, i.e., $\widetilde{P}_{se}^{CG}$ and $\widetilde{P}_{se}^{UD}$, the later is obviously better. This can be explained as follows: When $N$ is large, the received signal energy $\bar Z$ can be modeled as a Gaussian random variable. For complex Gaussian signal, the mean and variance of $\bar Z$ are $\mu_{gi}=N\sigma _{i}^2$ and $\sigma _{gi}^2=N\sigma _{i}^4$, respectively. Similarly, for deterministic signal, the mean and variance are given by $\mu_{pi}=N\sigma _{i}^2$ and  $\sigma_{pi}^2=N(2|h_i|^2P_sN_w+N_w^2)$, respectively. Now, let us consider the case $\sigma_0^2<\sigma_1^2$ for example, the error occurs when $\bar Z|\bar H_{0}> \bar Z|\bar H_{1}$, i.e., $\Delta Z=\bar Z|\bar H_{0}-\bar Z|\bar H_{1}>0$. It can be seen that the variance of $\Delta Z$ in the deterministic signal case is smaller. The BER can be represented by the shadow area in Fig. \ref{fig:fig112}, it is easy to see that BER of the deterministic signal is smaller than that of the complex Gaussian signal.}

\textcolor{blue}{\begin{figure}[htbp]
\centering
\includegraphics[width=2.5in]{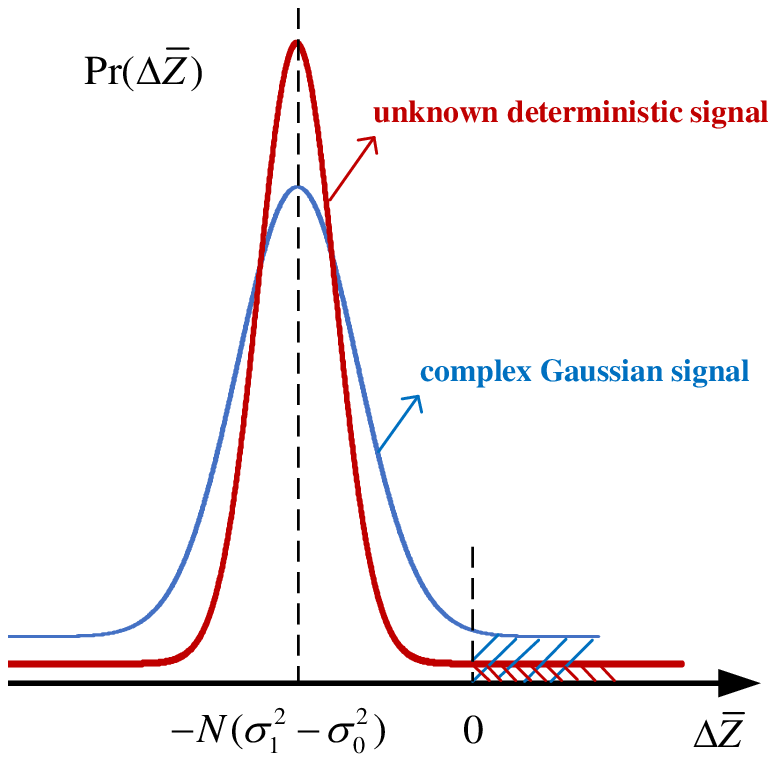}
\caption{The $\Pr (\Delta\bar Z)$ for the two scenarios when $\sigma_0^2<\sigma_1^2$}
\label{fig:fig112}
\end{figure}}

\section{Non-coherent Manchester detector}
The SeCoMC detector proposed in the previous section still needs to estimate the relationship of $\sigma_0^2$ and $\sigma_1^2$, which consumes some extra resources. Hence, in this section, we propose the NoCoMC detector based on differential Manchester coding which requires no training. Please note, the proposed NoCoMC detector differs from the non-coherent detectors proposed in \cite{2,5}, which still require some form of estimation.

\subsection{Complex Gaussian signal}

With differential Manchester coding, $d_k$ is determined by two adjacent symbols. As such, the detection of $d_k$ is based on two consecutive received signal vectors ${\bf{y}}_{{k - 1}}{\bf{y}}_{{k }}$. In this case, the ML detector can be obtained as
\begin{align} \label{g21}
L({\bf{y}}_{k-1}{\bf{y}}_k) = \frac{{\text{Pr}({\bf{y}}_{{k - 1}}{\bf{y}}_{{k }}|H_1)}}{{\text{Pr}({\bf{y}}_{{k - 1}}{\bf{y}}_{k}|H_0)}}
{\begin{array}{*{20}{c}}
{{H_1}}\\
 \gtrless\\
{{H_0}}
\end{array}}
1.
\end{align}
Now, let {$\hat H_{0}$ and $\hat H_{1}$} denote the hypotheses associated with $\hat d$=0 and $\hat d$=1, respectively, then we have
\begin{align}
&{\text{Pr}({\bf{y}}_{{k - 1}}{\bf{y}}_{{k }}|H_1)}\notag\\
&=\frac{1}{2}\text{Pr}({\bf{\hat y}}_{k-1}^a|\hat H_{0})\text{Pr}({\bf{\hat y}}_{k-1}^b|\hat H_{1})\text{Pr}({\bf{\hat y}}_{k}^a|\hat H_{1})\text{Pr}({\bf{\hat y}}_{k}^b|\hat H_{0})+ \nonumber \\
&\quad\frac{1}{2}\text{Pr}({\bf{\hat y}}_{k-1}^a|\hat H_{1})\text{Pr}({\bf{\hat y}}_{k-1}^b|\hat H_{0})\text{Pr}({\bf{\hat y}}_{k}^a|\hat H_{0})\text{Pr}({\bf{\hat y}}_{k}^b|\hat H_{1}), \label{gp1}\\
&{\text{Pr}({\bf{y}}_{{k - 1}}{\bf{y}}_{{k }}|H_0)}\notag\\
&=\frac{1}{2}\text{Pr}({\bf{\hat y}}_{k-1}^a|\hat H_{0})\text{Pr}({\bf{\hat y}}_{k-1}^b|\hat H_{1})\text{Pr}({\bf{\hat y}}_{k}^a|\hat H_{0})\text{Pr}({\bf{\hat y}}_{k}^b|\hat H_{1})
+\nonumber \\
&\quad\frac{1}{2}\text{Pr}({\bf{\hat y}}_{k-1}^a|\hat H_{1})\text{Pr}({\bf{\hat y}}_{k-1}^b|\hat H_{0})\text{Pr}({\bf{\hat y}}_{k}^a|\hat H_{1})\text{Pr}({\bf{\hat y}}_{k}^b|\hat H_{0}), \label{gp2}
\end{align}
Now, {defining
$\hat Z_{k}^{j}=||{\bf{\hat y}}_{{k}}^{j}|{|^2}$} and substituting (\ref{gp1}) and (\ref{gp2}) into (\ref{g21}), after some algebraic manipulations, we have
\begin{align} \label{gpj2}
(\hat Z_{k-1}^a - \hat Z_{k-1}^b)(\hat Z_{k}^a - \hat Z_{k}^b)
{\begin{array}{*{20}{c}}
{{H_0}}\\
 \gtrless\\
{{H_1}}
\end{array}}
0.
\end{align}

Interestingly, we see that the NoCoMC detector also resembles the energy detector. Different from the SeCoMC detector, which compares the energy level of the first and second half of a single symbol interval, the NoCoMC detector needs to jointly consider energy difference of two adjacent symbol intervals. According to (\ref{gpj2}), the decision region of the NoCoMC detector as illustrated in Fig. \ref{fig:fig4} can be defined as:
\begin{itemize}
\item Region $R1$: if $\hat Z_{k-1}^a-\hat Z_{k-1}^b\leq$ 0 and $\hat Z_{k}^a-\hat Z_{k}^b>$    0, then $d_k$=1;
\item Region $R2$: if $\hat Z_{k-1}^a-\hat Z_{k-1}^b>$    0 and $\hat Z_{k}^a-\hat Z_{k}^b\leq$ 0, then $d_k$=1;
\item Region $R3$: if $\hat Z_{k-1}^a-\hat Z_{k-1}^b\leq$ 0 and $\hat Z_{k}^a-\hat Z_{k}^b\leq$ 0, then $d_k$=0;
\item Region $R4$: if $\hat Z_{k-1}^a-\hat Z_{k-1}^b>$    0 and $\hat Z_{k}^a-\hat Z_{k}^b>$    0, then $d_k$=0.
\end{itemize}
\begin{figure}[htbp]
\centering
\includegraphics[width=3in]{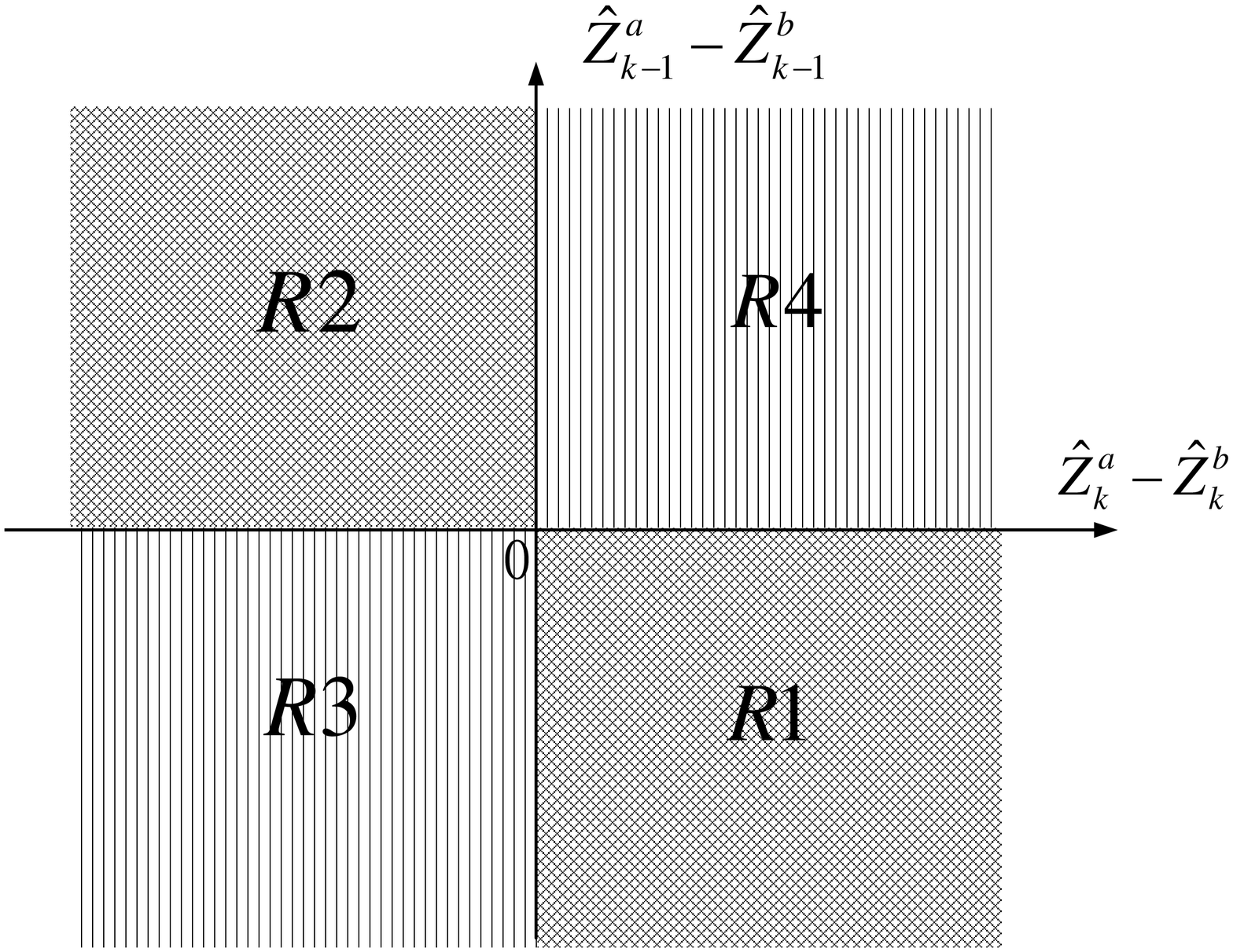}
\caption{Decision regions of the NoCoMC detector}
\label{fig:fig4}
\end{figure}

For symbol detection, every time the reader receives the signal vector ${\bf{y}}_k$, it compares the energy difference of the first and second half of the symbol interval, and stores the outcome in memory. Hence, the NoCoMC detector requires 1 bit additional memory. The algorithm for NoCoMC detector is shown as\\
\begin{algorithm}[h]
\caption{NoCoMC detector}
\begin{algorithmic} [1]
\Require The received signal vectors: [${\bf{y}}_0$; ${\bf{y}}_1$; $\cdots$ ${\bf{y}}_k$; $\cdots$ ${\bf{y}}_K$]
\Ensure The detected symbols: [$d_1$, $\cdots$, $d_k$ $\cdots$, $d_K$]
\State Get $A_0\triangleq \hat Z_{0}^a$ and $B_0\triangleq \hat Z_{0}^b$
\State For $k$ from 1 to $K$\\
    \quad compute $\hat Z_{k}^a$=$||{\bf{\hat y}}_{k}^a||^2$,$\hat Z_{k}^b$=$||{\bf{\hat y}}_{k}^b||^2$\\
     \quad if $k=1$\\
  \quad\quad if $(A_0-B_0)(\hat Z_{1}^a-\hat Z_{1}^b)<0$, then let $d_k=1$, else\\ \qquad \quad $d_k=0$， end if\\
  \quad\quad else if $(\hat Z_{k-1}^a-\hat Z_{k-1}^b)(\hat Z_{k}^a-\hat Z_{k}^b)<0$, then let $d_k=1$,\\ \qquad \quad else $d_k=0$, end if\\
  \quad end if\\
  \quad end for\\
  Return [$d_1$, $\cdots$, $d_k$ $\cdots$, $d_K$]
\end{algorithmic}
\end{algorithm}

To this end, we present a detailed performance analysis on the achievable BER of the proposed NoCoMC detector. We have the following important result:

\begin{theorem} \label{t02}
The BER of the NoCoMC detector with complex Gaussian signal is given by
\begin{align}  \label{g29}
P_{no}^{CG}&=2P_{se}^{CG}\left(1-P_{se}^{CG}\right),
\end{align}
\end{theorem}
where $P_{se}^{CG}$ is the BER of semi-coherent detector presented in Theorem \ref{t0}.
\proof
See Appendix \ref{ap4}.
\endproof
The BER of the NoCoMC detector given in (\ref{g29}) is actually quite intuitive. {The error occurs when only one of the adjacent symbol is incorrect, hence, the BER of the NoCoMC detector is $2P_{se}^{CG}(1-P_{se}^{CG})$. Moreover, since $P_{se}^{CG}\leq 1/2$, we have $P_{no}^{CG}\geq P_{se}^{CG}$}. Hence, the BER performance of NoCoMC detector is inferior to that of SeCoMC detector. To gain further insights, we now look into the asymptotic regime, where simple expression can be obtained.

\begin{theorem} \label{t2}
When $N$ is large, the BER of NoCoMC detector with complex Gaussian signal can be derived as
\begin{align}
\widetilde{P}_{no}^{CG}&=2\widetilde{P}_{se}^{CG}\left(1-\widetilde{P}_{se}^{CG}\right)\\
&=\frac{1}{2} - \frac{1}{2}{\sf erf}^2\left(\frac{{\sqrt N ||h_0|^2-|h_1|^2| }}{{\sqrt 2 \sqrt {{{(|{h_0}{|^2} + \frac{1}{\gamma })}^2} + {{(|{h_1}{|^2} + \frac{1}{\gamma })}^2}} }}\right),\label{g440}
\end{align}
where $\widetilde{P}_{se}^{CG}$ is the BER of semi-coherent detector presented in Theorem \ref{t1}.
\end{theorem}
\proof The result can be obtained in a similar fashion as in Theorem \ref{t02}.
\endproof

In the asymptotic high SNR regime, (\ref{g440}) reduces to
\begin{align} \label{g31}
\widetilde{P}_{no}^{CG} \approx \frac{1}{2} - \frac{1}{2}{\sf erf}^2\left(\frac{{\sqrt N ||h_0|^2-|h_1|^2| }}{{\sqrt 2 \sqrt {{{|{h_0}{|^4} }} + {{|{h_1}{|^4} }}} }}\right),
\end{align}
which is independent of the operating SNR, indicating the {existence} of an error floor.

\subsection{Unknown deterministic signal}
We now investigate the BER performance of NoCoMC detector with unknown deterministic signal, and we have the following important result:
\begin{theorem} \label{tp1}
When $N$ is large, the BER of NoCoMC detector with unknown deterministic signal is given by
\begin{align}
\widetilde{P}_{no}^{UD} &= \frac{1}{2} - \frac{1}{2}{\sf erf}^2\left(\frac{\sqrt N \left||h_1|^2 - |h_0|^2\right| }{2 \sqrt {\frac{|h_0|^2 +|h_1|^2}{\gamma}+ \frac{1}{\gamma^2 } }}\right).\label{g39}
\end{align}
\end{theorem}
When $\gamma$ approaches infinity, (\ref{g39}) can be simplified as
\begin{align} \label{g42}
\widetilde{P}_{no}^{UD}\approx \frac{1}{2} - \frac{1}{2}{\sf erf}^2\left(\frac{\sqrt {N\gamma} \left||h_1|^2 - |h_0|^2\right| }{2 \sqrt {|h_0|^2 +|h_1|^2}}\right).
\end{align}
Similar to the SeCoMC detector, we see that the BER of unknown deterministic signal does not settle and continues to decrease with the SNR.

\section{Numerical results} \label{svi}

In this section, we provide simulation results to validate the correctness of the analytical expressions and evaluate the performance of the proposed detectors. Since the distance between the source and the tag/reader is much longer than that between the tag and reader, we set ${h_{st}},{h_{sr}}\sim CN(0,1)$, ${h_{tr}}\sim CN(0,10)$, $K=30$, tag coefficient $\alpha$=0.5, and the AGWN follows $w(t)\sim CN(0,1)$. Without losing generality, $8$-PSK is used to model the unknown deterministic signal.

\begin{figure}[!ht]
\begin{minipage}[t]{0.48\linewidth}
\centering
\subfigure[Complex Gaussian signal] { \label{fig:fig6a}
\includegraphics[width=3in]{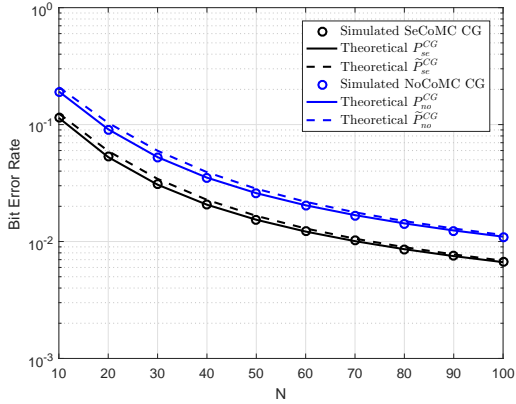}}
\end{minipage}
\\
\begin{minipage}[t]{0.48\linewidth}
\centering
\subfigure[$8$-PSK signal] { \label{fig:fig6b}
\includegraphics[width=3in]{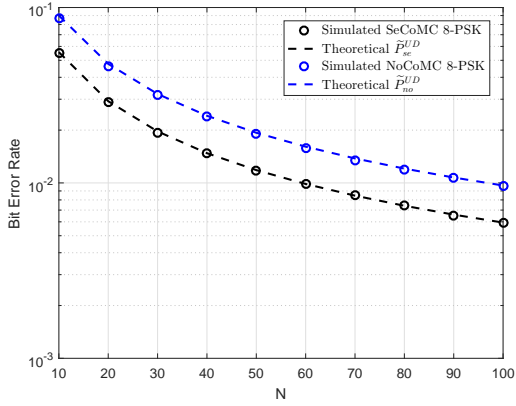}}
\end{minipage}
\caption{Impact of sampling rate $N$ on the BER performance}
\label{fig:fig6}
\end{figure}

Fig. \ref{fig:fig6} illustrates the BER performance of the proposed SeCoMC and NoCoMC detectors when $T=20$, {$q=0.5$}, $\gamma = 5$ dB, and $\mbox{RCD}=0.5$.
As can be readily observed, the analytical curves drawn according to $P_{se}^{CG}$, $P_{no}^{CG}$ match perfectly with the simulation curves, while the approximation curves drawn according to $\widetilde{P}_{se}^{CG}$, $\widetilde{P}_{no}^{CG}$, $\widetilde{P}_{se}^{UD}$, and $\widetilde{P}_{no}^{UD}$ are also reasonably accurate, especially for large $N$. In addition, for both detectors, the BER decrease as the sampling rate $N$ increase. Moreover, comparing the BER performance of the SeCoMC detector with the NoCoMC detector, we see that the BER performance of SeCoMC detector is always superior than that of the NoCoMC detector under the same condition, mainly due to error propagation between adjacent {differential Manchester codes}.

\begin{figure}[!ht]
\centering
\includegraphics[width=3in]{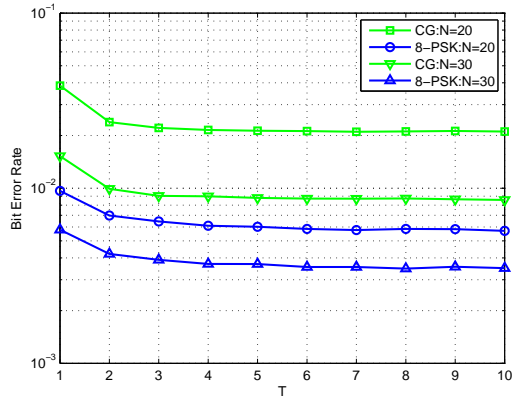}
\caption{Impact of training length on the BER performance}
\label{fig:fig5}
\end{figure}

Fig. \ref{fig:fig5} examines the impact of training length on the BER performance of the SeCoMC detector with $\gamma =10$ dB, $q=0.5$, and $\mbox{RCD}=0.5$. As can be observed, while it is rather intuitive that increasing the training length $T$ would lead to better BER performance, it is quite surprising that the performance gain becomes marginal when $T$ increases beyond 2. This is a rather encouraging outcome, indicating that regardless of the sampling rate, only a small fraction of the valuable time resource needs to be used for training.

\begin{figure}[!ht]
\centering
\includegraphics[width=3in]{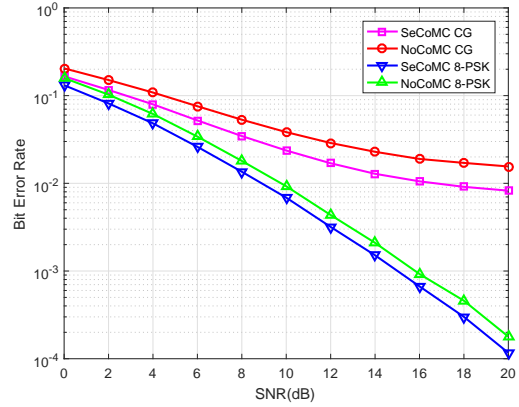}
\caption{Impact of SNR on the BER performance}
\label{fig:fig55}
\end{figure}

Fig.\ref{fig:fig55} compares the BER performance with the complex Gaussian and the 8-PSK signal with $q=0.5$, $N=20$, $T=2$, and $\mbox{RCD}=0.5$. For both detectors, we see that the BER with 8-PSK signal is always lower than that of the complex Gaussian signal. In the high SNR regime, the BER settles for the complex Gaussian signal while keeps falling for the 8-PSK signal as predicted by (\ref{g20}), (\ref{g37}), (\ref{g31}) and (\ref{g42}). Please note, similar trend has already been observed in \cite{1,2}.

\begin{figure}[!ht]
\centering
\includegraphics[width=3in]{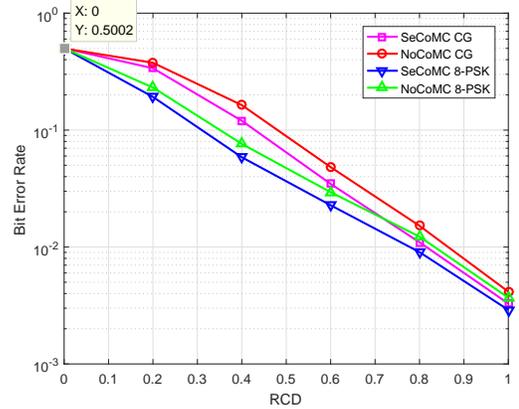}
\caption{Impact of RCD on the BER performance}
\label{fig:fig56}
\end{figure}

Fig.\ref{fig:fig56} illustrates the impact of RCD on the BER performance with $q=0.5$, $N=20$, $T=2$, and $\gamma=5$ dB. As can be readily observed, the BER curves of all four cases decrease with the increasing of RCD. This is intuitive, since both the SeCoMC and NoCoMC detectors are energy based. When the energy difference of the two hypotheses becomes more substantial, the detection performance improves. For the extreme case with $\mbox{RCD}=0$, we see that the BER of all four cases is nearly 0.5. This is expected, since $\mbox{RCD}=0$ implies identical energy of the two hypotheses, hence no reliable detection is possible, luckily, such scenario is considered unlikely.

\begin{figure}[htbp]
\centering
\includegraphics[width=3in]{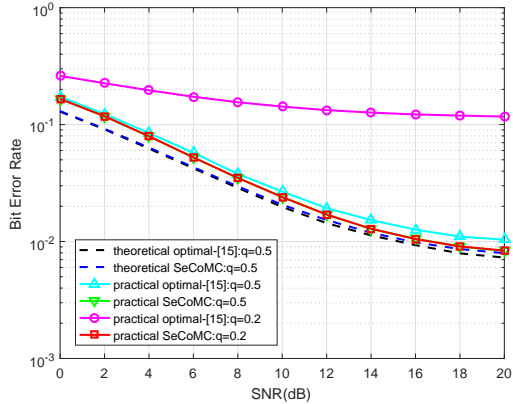}
\caption{BER comparison: SeCoMC detector VS. optimal detector \cite{1}}
\label{fig:fig8}
\end{figure}

Fig. \ref{fig:fig8} compares the BER performance of the SeCoMC detector with the optimal detector proposed in \cite{1} with complex Gaussian signal when $N=20$, $\mbox{RCD}=0.5$, and $T=2$ for different $q$. For fair comparison, the sampling rate of the two detectors is the same, i.e., $N$ samples are collected during each $\bar d$ interval in this paper, while $2N$ samples are collected during each $d$ interval in \cite{1}. Let us take a close look at the two dash curves in Fig. \ref{fig:fig8}, it can be seen that the curve associated with optimal detector in \cite{1} is slightly below in the high SNR regime. {The main reason is that there is some {information} loss due to the subtraction of the symbol energy in SeCoMC detector}. However, looking at the practical curves, we see that the proposed detector actually outperforms the optimal detector in \cite{1}, and the performance gain becomes more pronounced with higher SNR. The main reason lies in the fact that the detector in \cite{1} needs to estimate the decision threshold, whose accuracy is strictly limited by the estimation method and finite samples. In contrast, the proposed detector does not require the estimation of the decision threshold, hence is more robust in practice. Another disadvantage of the detector in \cite{1} is that, the detection process starts after the estimation of threshold, which incurs additional communication delay. Finally, we see that the distribution of information bits has a significant impact on the BER performance of the detectors in \cite{1}. For instance, the BER of the detector in \cite{1} is significantly higher when $q=0.2$. It is mainly because the decision threshold is estimated based on the assumption that ``0'' and ``1'' are equally probable.

\begin{figure}[!ht]
\centering
\includegraphics[width=3in]{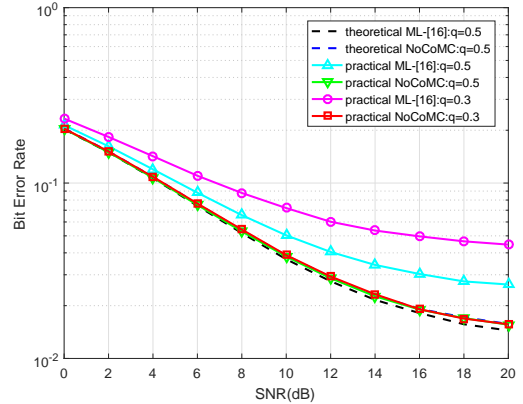}
\caption{BER comparison: NoCoMC detector VS. ML detector \cite{2}}
\label{fig:fig9}
\end{figure}

Fig. \ref{fig:fig9} compares the BER performance of the proposed NoCoMC detector and the ML detector in \cite{2} with complex Gaussian signal when $N=20$, $\mbox{RCD}=0.5$, and $T=2$ for different $q$. Similar to Fig. \ref{fig:fig8}, we observe that the theoretical BER performance of the NoCoMC detector is slightly higher than the ML detector, while the practical BER performance of the NoCoMC detector is substantially better than the ML detector. Please note, even if $q=0.5$, the practical BER performance of the ML detector deviates significantly away from the theoretical performance. The reason is that, given equally probable ``0'' and ``1'' in the original information bit sequence, the resulting bit sequence after differential modulation is no longer balanced, which degrades the estimation accuracy of the decision threshold. When $q=0.3$, the BER performance of the ML detector is even worse, while the BER performance of the proposed NoCoMC detector remains the same.
\section{Conclusion}
With Manchester and differential Manchester coding at the tag, this paper have proposed the SeCoMC and NoCoMC detectors to enable reliable detection. In additional, analytical closed-form expressions are derived for the BER of the system. Simulation results show that the BER performance of deterministic ambient signal is better, and the SeCoMC detector outperforms the NoCoMC detector. Moreover, the proposed detectors achieve superior BER performance compared with prior detectors in practice. Unlike the prior detectors, due to the unique property of the Manchester code, the proposed SeCoMC and NoCoMC detectors can work equally well with arbitrary distribution of the information bits. Furthermore, the proposed detectors enable immediate detection of each symbol, hence do not introduce any extra delay during the detection process.
\appendices
\section{Proof of Theorem 1 }\label{ap2}
Since we are dealing with the energy of the received signal, we find it more convenient to work in the real number domain. As such, we denote ${h_i}s[n] = s_i^R[n] + js_i^I[n]$ and $w[n] = {w^R}[n] + j{w^I}[n]$, where $s_i^R[n]$ and $s_i^I[n]$ represent the real and the imaginary parts of ${h_i}s[n]$, while ${w^R}[n]$ and ${w^I}[n]$ represent the real and the imaginary parts of $w[n]$, respectively. Then, the received signal energy can be expressed as
\begin{align} \label{g14}
{\bar Z}&=\sum\limits_{n = 0}^{N-1} {{{|{h_i}s[n] + w[n]|}^2}}\notag\\
&=\sum\limits_{n = 0}^{N-1} \left[(s_i^R[n]+{w^R}[n])^2+(s_i^I[n]+{w^I}[n])^2\right].
\end{align}

With complex Gaussian signal, $s_i^R[n]$ and $s_i^I[n]$ are both zero-mean Gaussian random variables which are independent of ${w^R}[n]$ and ${w^I}[n]$. Hence, $(s_i^R[n]+{w^R}[n])$ and $(s_i^I[n]+{w^I}[n])$ are zero-mean Gaussian random variables with variance $\frac{\sigma _{i}^2}{2}$. Therefore, $\bar Z$ follows the central chi-square distribution with $2N$ degree of freedom (DOF), i.e., $\bar Z \sim$ ${\chi}^2_{2N}$.

To this end, the PDF of $\bar Z$ under the hypothesis $\bar H_{i}$ is given by
\begin{align}
\text{Pr}(\bar Z|\bar H_{i}) = \frac{{\bar Z}^{N-1}e^{-\frac{\bar Z}{\sigma _{i}^{2}}}}{\Gamma(N)\sigma _{i}^{2N}},\quad \bar Z>0,
\end{align}
where $\Gamma(x)$ denotes the gamma function.

Now, assuming that the probabilities of hypothesis $H_1$ and $H_0$ are $q$ and $1-q$ $(0 \leq q \leq 1)$, respectively. Thus, if $\sigma _{0}^2> \sigma _{1}^2$, the corresponding BER can be computed by
\begin{align}
P_{se}^{CG} &= q\text{Pr}(\bar Z|\bar H_{0}<\bar Z|\bar H_{1}) + (1-q)\text{Pr}(\bar Z|\bar H_{0}< \bar Z|\bar H_{1})\label{g115}\\
&=\text{Pr}(\bar Z|\bar H_{0}< \bar Z|\bar H_{1})\label{g116}\\
&=\text{Pr}\left(\frac{\bar Z|\bar H_{0}}{\bar Z|\bar H_{1}}<1 \right).\label{g1116}
\end{align}
{Since the ratio of two independent chi-square random variables follows the F-distribution, (\ref{g1116}) can be evaluated as
\begin{align}
P_{se}^{CG}
&={\frac{\Gamma(2N)\sigma _{1}^{2N}}{N\Gamma^2(N)\sigma _{0}^{2N}}} \cdot {_2F_1}\left(N,2N;N+1;-\frac{\sigma _{1}^{2}}{\sigma _{0}^{2}}\right),\label{g117}
\end{align}
where ${_2F_1}(a, b ; c ; -x)$ denotes the Gauss hypergeometric function \cite{Table}.}

If $\sigma _{0}^2 < \sigma _{1}^2$, the corresponding BER can be similarly computed by
\begin{align} \label{g114}
P_{se}^{CG} &= {\frac{\Gamma(2N)\sigma _{0}^{2N}}{N\Gamma^2(N)\sigma _{1}^{2N}}} \cdot {_2F_1}\left(N,2N;N+1;-\frac{\sigma _{0}^{2}}{\sigma _{1}^{2}}\right).
\end{align}

To this end, combining (\ref{g117}) and (\ref{g114}) together yields the desired result.

\section{Proof of Theorem 2}\label{ap2n}
When $N$ is sufficiently large, the central limit theorem can be invoked to simplify the analysis. In particular, $\bar Z$ can be modeled by the Gaussian distribution with mean $\mu_{gi}=N\sigma _{i}^2$ and variance $\sigma _{gi}^2=N\sigma _{i}^4$. As such, the PDF of $\bar Z$ under the hypothesis ${\bar H_{i}}$ can be written as
\begin{align}
\text{Pr}(\bar Z|\bar H_{i}) = \frac{1}{{\sqrt {2\pi } \sigma _{gi}^{}}}{e^{ - \frac{{{{({\bar Z} - \mu _{gi})}^2}}}{{2\sigma _{gi}^2}}}}, \quad -\infty < \bar Z < +\infty.
\end{align}

{Then, the BER is equivalent to the difference of two normal random variables, which also follows normal distribution, being less than 0. Thus we have
\begin{align} \label{g13}
\widetilde{P}_{se}^{CG} &=\text{Pr}(\bar Z|\bar H_{0}- \bar Z|\bar H_{1}<0)\\
&=\frac{1}{2}{\sf erfc} \left(\frac{{\mu _{g0}^{} - \mu _{g1}^{}}}{{\sqrt {2(\sigma _{g1}^2 + \sigma _{g0}^2)} }}\right),\label{g1113}
\end{align}
where ${\sf erfc}(x)= 1 - {\sf erf}(x)$, ${\sf erf}(x) =\frac{2}{\sqrt{\pi} }\int_0^x {{e^{ - {t^2}}}} dt$.}

The case of $\sigma _{0}^2 < \sigma _{1}^2$ is similar, and combining two cases together yields
\begin{align}\label{g10009}
\widetilde{P}_{se}^{CG} &= \frac{1}{2}{\sf erfc} \left(\frac{|{\mu _{g1}^{} - \mu _{g0}^{}}|}{{\sqrt {2(\sigma _{g1}^2 + \sigma _{g0}^2)} }}\right).
\end{align}

Finally, substituting the appropriate system parameters into (\ref{g10009}) yields the desired results.

\section{Proof of Theorem 3 }\label{ap3}
With unknown deterministic signal, $s_i^R[n]$ and $s_i^I[n]$ are no longer Gaussian variables, instead they are constants satisfying
\begin{align}
s_i^R[n]^2+s_i^I[n]^2=|h_i|^2P_s = s_i^2.
\end{align}

In addition, $(s_i^R[n]+{w^R}[n])$ and $(s_i^I[n]+{w^I}[n])$ are Gaussian variables with means $s_i^R[n]$ and $s_i^I[n]$, respectively. Also, the variance of these two Gaussian variables is the same and given by $\sigma^2=\frac{N_w}{2}$. Therefore, $\bar Z$ follows the non-central chi-squared distribution with $2N$ DOF and non-central parameter $\lambda =\sum\limits_{n = 0}^{N-1}s_i^R[n]^2+s_i^I[n]^2=N|h_i|^2P_s$, i.e., $\bar Z \sim \chi_{2N}^{'2}(N|h_i|^2P_s)$.

Due to the complex PDF expression of non-central chi-square distribution, it is difficult to obtain closed-form BER expression. As such, we look into the asymptotic large $N$ regime. When $N$ is sufficiently large, the central limit theorem can be invoked. Hence, $\bar Z$ can be modeled by the normal distribution, with mean $\mu_{pi}$ and variance $\sigma_{pi}^2$ \cite{7}, where
\begin{align} \label{g32}
\mu_{pi}&=N(2\sigma^2+s_i^2)=N(|h_i|^2P_s+N_w),\nonumber\\
\sigma_{pi}^2&=N(4\sigma^4+4\sigma^2s_i^2)=N(2|h_i|^2P_sN_w+N_w^2).
\end{align}

To this end, following the same lines as in the proof of Theorem 2, the BER performance with deterministic signal can be obtained as
\begin{align}
\widetilde{P}_{se}^{UD}&= \frac{1}{2}{\sf erfc}\left(\frac{\sqrt N \left||h_1|^2 - |h_0|^2\right| }{2 \sqrt {\frac{|h_0|^2 +|h_1|^2}{\gamma}+ \frac{1}{\gamma^2 } }}\right).
\end{align}

\section{Proof of Theorem 4}\label{ap4}
Assuming the probabilities of $H_1$ and $H_0$ are $q$ and $1-q$, respectively. An incorrect detection takes place when {one of} the following two events occurs:
\begin{itemize}
\item When $d_k=1$, the reader decides $(\hat Z_{k-1}^a - \hat Z_{k-1}^b)(\hat Z_{k}^a - \hat Z_{k}^b) > 0$;
\item  When $d_k=0$, the reader decides $(\hat Z_{k-1}^a - \hat Z_{k-1}^b)(\hat Z_{k}^a - \hat Z_{k}^b) \leq 0$;
\end{itemize}

Hence, the BER can be computed by
\begin{align}\label{g25}
P_{no}^{CG} =& q\text{Pr}\left((\hat Z_{k-1}^a - \hat Z_{k-1}^b)(\hat Z_{k}^a - \hat Z_{k}^b) > 0|{H_1}\right)+\nonumber\\
&(1-q)\text{Pr}\left((\hat Z_{k-1}^a - \hat Z_{k-1}^b)(\hat Z_{k}^a - \hat Z_{k}^b) \leq 0|{H_0}\right),
\end{align}
where
\begin{multline}\label{g26}
\text{Pr}\left((\hat Z_{k-1}^a - \hat Z_{k-1}^b)(\hat Z_{k}^a - \hat Z_{k}^b) > 0|{H_1}\right) \\
=\text{Pr}(\hat Z_{k-1}^1 - \hat Z_{k-1}^0> 0)\text{Pr}(\hat Z_{k}^0 - \hat Z_{k}^1 > 0)+ \\
\text{Pr}(\hat Z_{k-1}^1 - \hat Z_{k-1}^0 < 0)\text{Pr}(\hat Z_{k}^0 - \hat Z_{k}^1) < 0),
\end{multline}
and
\begin{multline}\label{g27}
\text{Pr}\left((\hat Z_{k-1}^a - \hat Z_{k-1}^b)(\hat Z_{k}^a - \hat Z_{k}^b)\leqslant 0|{H_0}\right)  \\
=\text{Pr}(\hat Z_{k-1
}^1 - \hat Z_{k-1}^0 > 0)\text{Pr}(\hat Z_{k}^1 - \hat Z_{k}^0 \leqslant 0)+\\
\text{Pr}(\hat Z_{k-1}^1 - \hat Z_{k-1}^0\leqslant 0)\text{Pr}(\hat Z_{k}^1 - \hat Z_{k}^0 > 0).
\end{multline}

To proceed, we make the critical observation that $\hat Z_{k}^j$ has the same distribution as $\bar Z_{k}^j$. Hence, if $\sigma_0^2>\sigma_1^2$, according to (\ref{g13}), we have
\begin{align} \label{g28}
\text{Pr}(\hat Z_{{k}}^0 < \hat Z_{{k}}^1)=P_{se}^{CG}, \mbox{ and } \text{Pr}(\hat Z_{{k}}^0 \geqslant \hat Z_{{k}}^1)=1-P_{se}^{CG}.
\end{align}

To this end, combining (\ref{g25}), (\ref{g26}), (\ref{g27}) and (\ref{g28}) together, the BER of the NoCoMC detector can be computed by
\begin{align}
P_{no}^{CG}=&q\Big[(1-P_{se}^{CG})P_{se}^{CG}+P_{se}^{CG}(1-P_{se}^{CG})+\notag\\
&\quad P_{se}^{CG}(1-P_{se}^{CG})+(1-P_{se}^{CG})P_{se}^{CG}\Big]+\notag\\
&(1-q)\Big[(1-P_{se}^{CG})P_{se}^{CG}+P_{se}^{CG}(1-P_{se}^{CG})+\notag\\
&\qquad \qquad P_{se}^{CG}(1-P_{se}^{CG})+(1-P_{se}^{CG})P_{se}^{CG}\Big]\notag\\
=&2P_{se}^{CG}(1-P_{se}^{CG}).
\end{align}

For the case $\sigma_0^2<\sigma_1^2$, the same conclusion can be drawn.

\nocite{*}

\bibliographystyle{IEEE}

\begin{thebibliography}{10}

\bibitem{L.Atzori}
L. Atzori, A. Iera, and G. Morabito, ``The internet of things: A survey,'' {\em Computer Networks}, vol. 54, no. 15, pp. 2787--2805, Oct. 2010.


\bibitem{K.Han}
K. Han and K. Huang, ``Wirelessly powered backscatter communication networks: Modeling, coverage and capacity,"  {\em IEEE Trans. Wireless Commun.}, vol. 16, no. 4, pp. 2548-2561, Apr. 2017.


\bibitem{C.He2}
H. Guo, C. He, N. Wang, and M. Bolic, ``PSR: A novel high efficiency and easy-to-implement parallel algorithm for anticollision in RFID systems,'' {\em  IEEE Trans. Ind. Informat.}, vol. 12, no. 3, pp. 1134-1145, Mar. 2016.

\bibitem{C.He4}
C. He, Z. J. Wang, C. Miao, and V. C. M. Leung, ``Block-level unitary query: Enabling orthogonal-like space-time code with query diversity for MIMO backscatter RFID,'' {\em IEEE Trans. Wireless Commun.}, vol. 15, no. 3, pp. 1937-1949, Mar. 2016.


\bibitem{DMDobkin}
D. M. Dobkin, {\em The RF in RFID: Passive UHF RFID in Practice}, Newnes: Elsevier, 2th edition, 2008.


\bibitem{30}
J. Kimionis, A. Bletsas, and J. N. Sahalos, ``Increased range bistatic scatter radio,'' {\em IEEE Trans. Commun.}, vol. 62, no. 3, pp. 1091--1104, Mar. 2014.

\bibitem{3}
V. Liu, A. Parks, V. Talla, S. Gollakota, D. Wetherall, and J. R. Smith, ``Ambient backscatter: Wireless communication out of thin air,'' in {\em Proc. ACM SIGCOMM'13}, Hong Kong, China, Aug. 2013, pp. 39--50.

\bibitem{wifi}
P. Zhang, D. Bharadia, K. Joshi, and S. Katti, ``HitchHike: Practical backscatter using commodity WiFi,'' in {\em Proc. ACM ENSSCR'16}, New York, NY, USA, Nov. 2016, pp. 259-271.

\bibitem{Kellogg2015Wifi}
B. Kellogg, A. Parks, S. Gollakota, J. R. Smith, and D. Wetherall, ``Wi-fi backscatter: Internet connectivity for RF-powered devices,'' in {\em Proc. ACM SIGCOMM'14}, Chicago, IL, USA, Aug. 2014, pp. 607-618.

\bibitem{BackFi}
D. Bharadia, K. R. Joshi, M. Kotaru, and S. Katti, ``BackFi: High throughput WiFi backscatter,'' {\em SIGCOMM Comput. Commun. Rev.}, vol. 45, no. 4, pp. 283--296, Aug. 2015.

\bibitem{18}
A. Parks, A. Liu, S. Gollakota, and J. R. Smith, ``Turbocharging ambient backscatter communication,'' in {\em {Proc. ACM} SIGCOMM'14}, Chicago, USA, Aug. 2014, pp. 17--22.


\bibitem{YLiu}
Y. Liu, G. Wang, Z. Dou, and Z. Zhong, ``Coding and detection schemes for ambient backscatter communication systems,'' {\em IEEE Access}, vol. 5, no. 99, pp. 4947--4953, Mar. 2017.

\bibitem{DDarsena}
D. Darsena, G. Gelli, and F. Verde, ``Modeling and performance analysis of wireless networks with ambient backscatter devices,''  {\em IEEE Trans. Commun.}, vol. 65, no. 4, pp. 1797--1814, Apr. 2017.


\bibitem{HoangDT}
D. T. Hoang, D. Niyato, W. Ping, I. K. Dong, and H. Zhu, ``Ambient backscatter: A new approach to improve network performance for RF-powered cognitive radio networks,'' {\em IEEE Trans. Commun.}, vol. 65, no. 9, pp. 3659--3674, Sep. 2017.


\bibitem{1}
J. Qian, F. Gao, G. Wang, S. Jin, and H. Zhu, ``Semi-coherent detection and performance analysis for ambient backscatter system,'' {\em IEEE Trans. Commun.}, vol. 65, no. 12, pp. 5266-5279, Dec. 2017.

\bibitem{2}
J. Qian, F. Gao, G. Wang, S. Jin, and H. Zhu, ``Noncoherent detections for ambient backscatter system,'' {\em IEEE Trans. Wireless Commun.}, vol. 16, no. 3, pp. 1412--1422, Mar. 2017.

\bibitem{5}
G. Wang, F. Gao, R. Fan, and C. Tellambura, ``Ambient backscatter communication systems: Detection and performance analysis,'' {\em IEEE Trans. Commun.}, vol. 64, no. 11, pp. 4836--4846, Nov. 2016.


\bibitem{45}
Y. H. Chen et al., ``A novel anti-collision algorithm in RFID systems for identifying passive tags," {\em IEEE Trans. Ind. Informat.}, vol. 6, no. 1, pp. 105-121, Feb. 2010.




\bibitem{44}
S. Prathipati, V. K. Samana, and J. U. Kidav, ``A 45nm FM0/Manchester code generator with PT logic running at 4GHz for DSRC applications," in {\em Proc. IEEE ICCC'15}, Trivandrum, India, Nov. 2015, pp. 558--562.

\bibitem{46}
Y. C. Lai, L. Y. Hsiao, and B. S. Lin, ``Optimal slot assignment for binary tracking tree protocol in RFID tag identification," {\em IEEE/ACM Trans. Networking}, vol. 23, no. 1, pp. 255--268, Feb. 2015.

\bibitem{47}
P. Ishwerya, V. N. Kumar, and G. Lakshminarayanan, ``An efficient digital baseband encoder for short range wireless communication applications," in {\em Proc. IEEE ICEEOT'16}, Chennai, India, Mar. 2016, pp. 2775--2779.

\bibitem{43}
A. I. Barbero, E. Rosnes, G. Yang, and O. Ytrebus, ``Near-Field passive RFID communication: Channel model and code desigh,'' {\em IEEE Trans. Commun.}, vol. 62, no. 5, pp. 1716--1727, May 2014.
{\bibitem{MSimon}
K. Finkenzeller, {\em RFID Handbook: Fundamentals and Applications in Contactless Smart Cards, Radio Frequency Identification and Near-Field Communication}, New York: Wiley, 3th edition, 2010.}

\bibitem{48}
H. Urkowitz, ``Energy detection of unknown deterministic signals," {\em Proceedings of the IEEE}, vol. 55, no. 4, pp. 523--531, Apr. 1967.

\bibitem{35}
M. K. Simon, {\em Proability Distributions Involving Gaussian Random Variales: A Handbook for Engineers, Scientists and Mathematicians}, New York: Springer, 2006. 

\bibitem{36}
A. Papoulis and S. U. Pillai, {\em Probability, Random Variables and Stochastic Processes}, New York: McGraw-Hill, 4th edition, 1991.

\bibitem{7}
J. G. Proakis, {\em Digital Communications}, New York: McGraw-Hill, 5th edition, 2007.

\bibitem{Table}
I. S. Gradshteyn and I.M.Rzhik, {\em Table of Integrals, Series, and Products}, California: Elsevier, 7th edition, 2007.
\end{thebibliography}
\begin{footnotesize}

\end{footnotesize}

\end{document}